\documentclass[runningheads]{llncs}
\usepackage{hyperref}
\usepackage{amsmath}
\usepackage{amssymb}
\usepackage{enumerate}
\usepackage[all]{xy}
\allowdisplaybreaks[4]

\newcommand{\M}{\mathfrak{M}}
\newcommand{\N}{\mathfrak{N}}

\newcommand{\A}{\mathcal{A}}
\newcommand{\K}{\mathcal{K}}
\renewcommand{\P}{\mathcal{P}}

\newcommand{\U}{\mathbb{U}}
\newcommand{\V}{\mathbb{V}}
\newcommand{\W}{\mathbb{W}}

\newcommand{\warning}{\textsf{Warning}}
\newcommand{\Sfive}{\texttt{S5}}

\renewcommand{\phi}{\varphi}

\begin{document}

\title{\texorpdfstring{Static Knowledge vs. Dynamic Argumentation\\\sf\large A Dual Theory Based on Kripke Semantics}{Static Knowledge vs. Dynamic Argumentation}}
\titlerunning{Static Knowledge vs. Dynamic Argumentation}
\author{\tt Xinyu Wang\and Momoka Fujieda}
\authorrunning{X. Wang\and M. Fujieda}
\institute{\tt School of Information Science\\Japan Advanced Institute of Science and Technology}

\maketitle

\begin{abstract}
This paper establishes a dual theory about knowledge and argumentation. Our idea is rooted at both epistemic logic and argumentation theory, and we aim to merge these two fields, not just in a superficial way but to thoroughly disclose the intrinsic relevance between knowledge and argumentation. Specifically, we define epistemic Kripke models and argument Kripke models as a dual pair, and then work out a two-way generation method between these two types of Kripke models. Such generation is rigorously justified by a duality theorem on modal formulae's invariance. We also provide realistic examples to demonstrate our generation, through which our framework's practical utility gets strongly advocated. We finally propose a philosophical thesis that knowledge is essentially dynamic, and we draw certain connection to Maxwell's demon as well as the well-known proverb ``knowledge is power''.

\end{abstract}

\section{Introduction}\label{sec.int}

As we learn from Newtonian mechanics, motion is relative rather than absolute \cite{Kittel73}, and quantum mechanics further tells us about wave-particle duality \cite{Wichmann71}. Nature proves to be highly symmetric and unified. In epistemic logic, agents' instantaneous knowledge is \textit{statically} represented via current possible worlds \cite{vanDitmarsch08}, viz. implicitly through the equivalence class of all the indistinguishable worlds \textit{at this moment}. Whereas in argumentation theory, the first-class citizens are chronological arguments that are \textit{dynamically} articulated by different agents so as to attack each other \cite{Dung95}, explicitly composing a sequential history of a tit-for-tat debate \textit{during a period of time}.

Following some pioneering attempts to describe argumentation theory with modal logic \cite{Grossi10,Grossi11}, a number of recent studies continue to explore assorted means of combining these two dimensions, namely agents' knowledge and argumentation, into one compound Kripke-like model \cite{Schwarzentruber12,Caminada15,Proietti21}. Nevertheless as far as we can see, those cursory treatments mostly scratch the surface, for under their settings generally speaking, knowledge and argumentation \textit{ipso facto} reside within two distinct strata of the model which are evaluated in completely separate ways, and so they are totally free to vary independently without any innate relevance. Such kind of ad hoc solution fails to fully conform to our intuition. Actually just from everyday common sense, we can assuredly comprehend that on the one hand, someone (supposedly honest) argues for something, most simply because he knows/believes what he argues for is true; on the other hand even more noticeably, we always infer others' knowledge/belief based on their arguments and actions, i.e. what they have said and done --- essentially speaking, it is exactly the \textit{only} way through which we may learn about others' knowledge. After all, in order to find out what others really know, it would help very little even if we were allowed to biologically dissect others' brains and examine their nerve cells! Therefore of all things, knowledge and argumentation are fundamentally intertwined.

In this paper, we propose a dual theory which simultaneously deals with both static knowledge and dynamic argumentation, from an intrinsic perspective that exquisitely unveils the close mutual interference between these two concepts. Specifically speaking, we shall define two different types of Kripke models, one for epistemic whilst another for argument, and then we shall demonstrate the two-way transformation between the two types. This is to say literally, given some epistemic-type Kripke model, we are able to compute out its corresponding argument-type Kripke model, and vice versa, the epistemic-type Kripke model can be generated from the argument-type Kripke model as well. In order to justify our framework, not only shall we elaborate on the philosophical intuition behind our definitions, but we shall also rigidly prove a duality theorem so that readers will get positively convinced of our approach's plausibility. Our novel disposal of the relationship between knowledge and argumentation incisively sheds light on both philosophical depth and mathematical complicacy, furthermore, many particular aspects of our method's superiority will be vividly exhibited as we leisurely walk through pragmatic examples of our logic's application. At last we shall briefly discuss a few possible expansions and some other relevant topics.

The rest of this paper is organized as follows. Section~\ref{sec.pre} defines preliminaries of our framework, in particular two types of Kripke models, namely the epistemic one and the argument one. Section~\ref{sec.k2a} illustrates how we can generate the corresponding argument Kripke model given the current epistemic Kripke model. Section~\ref{sec.a2k} subsequently establishes the reverse generation, i.e. from the current argument Kripke model to its corresponding epistemic Kripke model. Since this latter direction of generation may look a bit more complex, in the next Section~\ref{sec.exa} we at length analyze a truly practical example of a virtual debate between two candidates for the mayor, which significantly affirms our logic's affluent utility in the real world. Section~\ref{sec.dua} then is devoted to technical issues and ends up with a duality theorem on the two-way generation between two types of Kripke models, hence formally justifying our dual definition. Section~\ref{sec.dis} further discusses several relevant expansions and applications of our framework, especially invoking the famous thought experiment of Maxwell's demon so as to aver that knowledge is a dynamic notion in essence. Section~\ref{sec.con} finally concludes this paper and suggests possible directions of future work.

\section{Preliminaries}\label{sec.pre}

Fix a nonempty set of propositions $P$ and a nonempty set of agents $I$.
\begin{definition}[Language]\label{def.lan}
There are two types of well formed formulae: an epistemic formula $\phi$ and an argument formula $\xi$. Respectively, we define them by the following Backus–Naur forms, where $p\in P$, $i\in I$:
\begin{align*}
    \phi::=p\mid\neg\phi\mid(\phi\land\phi)\mid\K_i\phi & & \xi::=i\mid\neg\xi\mid(\xi\land\xi)\mid\A_p\xi
\end{align*}
\end{definition}
\begin{remark}
Readers may immediately notice the obvious duality between the two Backus–Naur forms in Definition~\ref{def.lan}. Indeed, as this paper aims at proposing a dual theory based on Kripke semantics, it is hardly a surprise that pairs of epistemic notions and argument notions will appear in parallel throughout the whole paper. Moreover, such symmetry also functions as a signpost to help readers smoothly navigate this paper.
\end{remark}
The intuitive readings for epistemic formulae are just as usual, e.g., epistemic formula $\K_i\phi$ means ``at the current world, agent $i$ knows $\phi$''. For argument formulae, intuitively, argument formula $i$ means ``the current argument is available to agent $i$'', and argument formula $\A_p\xi$ means ``all the arguments that attack the current argument with respect to proposition $p$ satisfy $\xi$''. Why we stipulate these intuitive readings for argument formulae will gradually become clear to readers as we continue unfolding our entire framework, but for now, if readers cannot help feeling somehow uneasy about our argument formulae or wondering whether they are really useful at all, then readers are strongly advised to take a quick look at the later Subsection~\ref{sub.syn}, where we concisely yet convincingly provide further elaboration on our dual syntax, for example what argument formula $\A_p\neg i$ intuitively means, and also mention some possible syntactic expansions.

Anyway, the syntactic respect being as such, on the other side, the semantic respect also turns out to be quite standard Kripke models:
\begin{definition}[Epistemic Kripke Model]\label{def.ekm}
An epistemic Kripke model is a tuple $\M=(S,E,F)$ where:
\begin{itemize}
    \item $S$ is a nonempty set of possible worlds.
    \item $E:I\to\P(S\times S)$ is a function that assigns an equivalence epistemic relation $E(i)\subseteq S\times S$ to each agent $i\in I$.
    \item $F:P\to\P(S)$ is a valuation function.
\end{itemize}
\end{definition}
\begin{definition}[Argument Kripke Model]\label{def.akm}
An argument Kripke model is a tuple $\N=(\W,A,G)$ where:
\begin{itemize}
    \item $\W$ is a nonempty set of possible arguments.
    \item $A:P\to\P(\W\times\W)$ is a function that assigns an irreflexive attack relation $A(p)\subseteq\W\times\W$ to each proposition $p\in P$.
    \item $G:I\to\P(\W)$ is an availability function.
\end{itemize}
\end{definition}
\begin{remark}
In accordance to common practice of epistemic logic, Definition~\ref{def.ekm} requires $E(i)$ to be an equivalence relation, i.e., epistemic Kripke models are \Sfive. But what restrictions should argument Kripke models conform to? We speculate that at least, any argument must not attack itself, and therefore Definition~\ref{def.akm} requires relation $A(p)$ to be irreflexive.
\end{remark}
An argument Kripke model in Definition~\ref{def.akm} is intuitively interpreted as the following:
\begin{itemize}
    \item For any $U,V\in\W$ and $p\in P$, $(U,V)\in A(p)$ means ``argument $V$ attacks argument $U$ with respect to proposition $p$''. \warning: $(U,V)\in A(p)$ means ``$V$ attacks $U$'', \textit{not} ``$U$ attacks $V$''.
    \item For any $U\in\W$ and $i\in I$, $U\in G(i)$ means ``argument $U$ is available for agent $i$ to utter out''. As we see, $\W$ includes all \textit{epistemically} possible arguments, such as ``it must be over 7am now, since it is so bright outside'' and ``it must be before 7am now, since there are few people on the street''. Nonetheless, for any rational agent $i$ at certain fixed world, apparently both arguments cannot be \textit{practically} accessible to him at the same time, otherwise it would imply absurdity. Therefore, the subset $G(i)\subseteq\W$ denotes all the arguments currently available to agent $i$.
\end{itemize}
\begin{remark}
Here we would like to emphasize once again, as our above ``\warning'' asserts, $(U,V)\in A(p)$ intuitively means ``$V$ attacks $U$'' rather than ``$U$ attacks $V$''. So why do we adopt such a ``reverse'' definition? Certainly, the whole argument model as a graph structure would just carry the very same information even if we had defined it in the opposite way, however as a Kripke model, the direction of the binary attack relation $A(p)$ matters really huge \cite{vanBenthem82}. Therefore to tell the truth, it should just be our chosen direction rather than the reverse that can comply well with other parts of this paper. All the relevant mathematical details will formally manifest such a fact along with our paper's progressing.
\end{remark}
\begin{definition}[Semantics]\label{def.sem}
For epistemic formulae, given an epistemic Kripke model $\M=(S,E,F)$ and a possible world $s\in S$:
\begin{align*}
    \M,s\vDash p\iff & s\in F(p)\\
    \M,s\vDash\neg\phi\iff & \text{\rm{not }}\M,s\vDash\phi\\
    \M,s\vDash\phi\land\psi\iff & \M,s\vDash\phi\text{\rm{ and }}\M,s\vDash\psi\\
    \M,s\vDash\K_i\phi\iff & \text{\rm{for all }}t\in S\text{\rm{ such that }}(s,t)\in E(i),\M,t\vDash\phi
\end{align*}
For argument formulae, given an argument Kripke model $\N=(\W,A,G)$ and a possible argument $U\in\W$:
\begin{align*}
    \N,U\vDash i\iff & U\in G(i)\\
    \N,U\vDash\neg\xi\iff & \text{\rm{not }}\N,U\vDash\xi\\
    \N,U\vDash\xi\land\zeta\iff & \N,U\vDash\xi\text{\rm{ and }}\N,U\vDash\zeta\\
    \N,U\vDash\A_p\xi\iff & \text{\rm{for all }}V\in\W\text{\rm{ such that }}(U,V)\in A(p),\N,V\vDash\xi
\end{align*}
\end{definition}
Semantics in Definition~\ref{def.sem} is just routine, based on which readers can easily verify that our intuitive readings for both epistemic and argument formulae are indeed appropriate.

Next, in the following Sections \ref{sec.k2a} and \ref{sec.a2k}, we are going to rigorously define a two-way generation between these dual types of Kripke models, which constitutes the backbone of this paper.

\section{From Knowledge to Argumentation}\label{sec.k2a}

Our first task is to generate an argument Kripke model $\N_g$, given an epistemic Kripke model $\M=(S,E,F)$ and the current world $s\in S$. Thus more precisely, $\N_g$ depends both on $\M$ and on $s$ and may be more explicitly denoted as $\N_g(\M,s)$ or $\N_g(s)$, but we shall anyway denote simply as $\N_g$ as long as $\M$ and $s$ are clear from the context without any possible confusion.

Now that we are provided with all the epistemically possible worlds, what are possible arguments then? In fact ontologically, an argument advocates an opinion about the world, asserting what the world is like or what the world is not like, such as ``the earth is getting warmer and warmer'' or ``unicorn does not exist''. Therefore technically speaking, an argument is essentially a nonempty subset of possible worlds, claiming that our current world lies within this subset. Of course, an empty set stands for a self-contradictory argument, and we can safely ignore such a case under our innocent assumption that all agents are rational.

If arguments are subsets of possible worlds, then how about attack relations between arguments? In fact with respect to any proposition $p$, an argument (i.e., a nonempty subset of possible worlds) must hold one out of three distinct types of attitudes:
\begin{enumerate}
    \item\label{ite.true}The argument asserts $p$, i.e., all the possible worlds in the subset satisfy $p$.
    \item\label{ite.false}The argument asserts $\neg p$, i.e., all the possible worlds in the subset falsify $p$.
    \item\label{ite.both}The argument does not tell whether $p$, i.e., some possible worlds in the subset satisfy $p$ while the others falsify $p$.
\end{enumerate}
Hence, several different versions of attack relations might look plausible:
\begin{enumerate}[(i)]
    \item Weak attack relation: only case~\ref{ite.true} and case~\ref{ite.false} attack each other.
    \item Medium attack relation: case~\ref{ite.true} and case~\ref{ite.false} attack each other, and both of them attack case~\ref{ite.both}.
    \item Strong attack relation: any two different cases attack each other.
\end{enumerate}
While all the above versions of attack relations may be quite reasonable by one way or another, in this paper, we choose to stick to the medium attack relation. Our intuition is such: different arguments possess different strength. Some arguments are stronger while others are weaker, and generally speaking, stronger arguments attack weaker arguments. The universal set $S$ itself is a trivial argument, as it just says that every possible world is possible and thus essentially argues nothing at all, therefore, it must be one of the weakest arguments and is supposed to be attacked by many other stronger arguments. Assume that with respect to some proposition $p$, the trivial argument $S$ is case~\ref{ite.both}, (which indeed seems most likely to be,) then we expect it to be attacked by arguments of case~\ref{ite.true} and case~\ref{ite.false}, but not vice versa. The above explanation accounts for our preference for the medium attack relation. Besides, only the medium attack relation is not necessarily symmetric, a crucial feature which undoubtedly makes the medium attack relation both more philosophically profound and more technically interesting.

By the way, we also point out in the later Subsection~\ref{sub.sem} that Dung's classic argument framework \cite{Dung95} is probably closer to the weak attack relation. Nevertheless for our project here, the weak attack relation looks pitifully too weak: case~\ref{ite.both} is just radically ruled out of the game so that much information about the original epistemic Kripke model will get lost.

Finally, we also have to figure out the last part in the generated argument Kripke model $\N_g=(\W,A,G)$, namely about the availability function $G$. Here we simply assume every agent is rational and honest, so that his available arguments should not exceed his own knowledge, that is to say, an argument (i.e., a nonempty subset of possible worlds $U\subseteq S$) is available to an agent $i\in I$ if and only if $U$ is a superset of $i$'s current epistemic equivalence class. Putting all the above discussion together, we obtain the following Definition~\ref{def.k2a}:
\begin{definition}[Generated Argument Kripke Model]\label{def.k2a}
Given an epistemic Kripke model $\M=(S,E,F)$ with a possible world $s\in S$, the generated argument Kripke model is $\N_g=(\W,A,G)$, where:
\begin{itemize}
    \item $\W=\P(S)\setminus\emptyset$.
    \item For any $p\in P$ and any $U,V\in\W$, $(U,V)\in A(p)$ iff $\exists t\in U\forall r\in V$, $(t\in F(p)\land r\notin F(p))\lor(t\notin F(p)\land r\in F(p))$.
    \item For any $i\in I$, $G(i)=\{U\in\W\mid\forall t\in S,(s,t)\in E(i)\to t\in U\}$.
\end{itemize}
\end{definition}
It is easy to see that $\W\neq\emptyset$ and that $A(p)$ is irreflexive, thus $\N_g$ in Definition~\ref{def.k2a} is indeed a well-defined argument Kripke model. In order to vividly demonstrate this generating process, we provide a simple example below:
\begin{example}\label{exa.sim}
Suppose there are two different propositions $p,q\in P$ and two different agents $a,b\in I$.

The original epistemic Kripke model is given as the following $\M$ with three possible worlds $S=\{s_1,s_2,s_3\}$, where any proposition not shown at a possible world is falsified there by default. The present world is $s_2$, therefore both $p$ and $q$ actually hold, but agent $a$ only knows the former fact and agent $b$ the latter.
$$
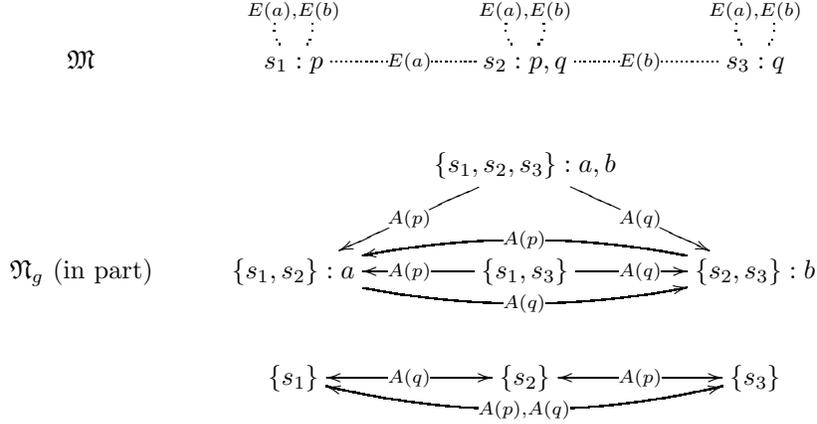
\xymatrix{
\M & s_1:p\ar@(ur,ul)@{.}|{E(a),E(b)}\ar@{.}[r]|{E(a)} & s_2:p,q\ar@(ur,ul)@{.}|{E(a),E(b)} & s_3:q\ar@(ur,ul)@{.}|{E(a),E(b)}\ar@{.}[l]|{E(b)}\\
& & \{s_1,s_2,s_3\}:a,b\ar[dl]|{A(p)}\ar[dr]|{A(q)}\\
\N_g\text{ (in part)} & \{s_1,s_2\}:a\ar@/_1pc/[rr]|{A(q)} & \{s_1,s_3\}\ar[l]|{A(p)}\ar[r]|{A(q)} & \{s_2,s_3\}:b\ar@/_1pc/[ll]|{A(p)}\\
& \{s_1\}\ar@{<->}[r]|{A(q)}\ar@{<->}@/_1pc/[rr]|{A(p),A(q)} & \{s_2\} & \{s_3\}\ar@{<->}[l]|{A(p)}
}$$
We can then compute out the generated argument Kripke model as the above $\N_g$, where there are seven possible arguments in total as the seven nonempty subsets of $S$, however in fact, just a part of $\N_g$ is shown for the sake of neatness. Specifically speaking, any attack relations between the upper four arguments and the lower three are omitted, but anyway readers can easily complete the whole $\N_g$ if needed. An argument is available to agent $a$ or $b$ only if it gets explicitly marked out in the above partial diagram of $\N_g$, and thus we can perceive that an (honest and rational) agent is only able to argue for some superset of his current epistemic equivalence class of possible worlds.

From the generated argument Kripke model $\N_g$, we can also easily see that the possible argument $\{s_1,s_2\}$ is available to agent $a$ and has the power to attack many other arguments with respect to proposition $p$, exactly because agent $a$ knows $p$ in the original epistemic Kripke model $\M$ at the current possible world $s_2$, and vise versa for the possible argument $\{s_2,s_3\}$, agent $b$ and proposition $q$. \textit{Ergo scientia est potentia}. In sum, the anticipated interaction between knowledge and argumentation indeed coincides with our intuition pretty well.
\end{example}

\section{From Argumentation to Knowledge}\label{sec.a2k}

Our second task is to generate an epistemic Kripke model $\M_g$. In parallel with Definition~\ref{def.k2a} in Section~\ref{sec.k2a}, here our start point should be an argument Kripke model $\N=(\W,A,G)$ together with the current argument $U\in\W$. But what does the ``current'' argument mean? We understand it as the last argument in a realistic argumentation process consisting of a chronological series of successive arguments, which is also ordinarily regarded as ``the currently winning argument'', exactly like Dung points out in \cite{Dung95}:
\begin{quote}
The way humans argue is based on a very simple principle which is summarized succinctly by an old saying: ``\textit{The one who has the last word laughs best}''.
\end{quote}
So suppose we are now provided with an argument Kripke model and the current argument, what on earth should be all the epistemically possible worlds? It could look like a merely impossible mission at the very first glance. Nevertheless, considering how we manage to derive argumentation from knowledge in Section~\ref{sec.k2a}, conceptually speaking, if our knowledge determines our argumentation, then our argumentation must also reflect our knowledge the other way round. Also from everyday experience, it is actually quite common and natural for us to infer others' epistemic states based on their utterances, namely, arguments. Therefore, just as nonempty subsets of possible worlds can represent possible arguments, let us in turn try to seek what kinds of subsets of possible arguments can represent possible worlds.

As it is assumed that all agents are rational, roughly speaking, all the possible arguments available to certain agent should form a \textit{consistent} subset, however not necessarily \textit{complete}. Instead, very similar to the canonical model method that gets extensively utilized in completeness proof of modal logic \cite{Goldblatt92}, the core philosophical idea of our treatment here is to regard \textit{maximal consistent} subsets of possible arguments as possible worlds. This subset of possible arguments are exactly all those arguments that actually hold true at the corresponding possible world, hence it must be consistent, and also once it is maximal, this subset contains so many possible arguments that they completely cover all aspects of a world's characteristics, thus uniquely determining \textit{the} one possible world.

For now, in order to formalize the above intuitive idea, one subtle technical issue remains to be solved with care: we still lack a rigid definition of so-called ``consistent'' subsets of possible arguments. After all, recall that now we are directly provided with an argument Kripke model, which is never guaranteed to have been actually generated from some epistemic Kripke model; in fact based on Definition~\ref{def.akm}, we know nothing but attack relations are irreflexive in argument Kripke models. Hence, we have to define a very general notion of ``consistency'' through very general mathematical tools.

Nevertheless, some sort of hint drawn from our established methodology of the generated argument Kripke model may still prove to be quite enlightening. For a figurative clue let us refer back to Example~\ref{exa.sim}, where the generated argument Kripke model $\N_g$ very similarly resembles the power set lattice of $S$ \cite{Artin91}, so apparent an evidence that drives us to naturally think of the bijective correspondence between any element in a set and its principal ultrafilter, which is exactly a subset of the power set. Such a broad orientation via ultrafilters seems quite promising, however as we can confirm, an argument Kripke model does not necessarily form a Boolean algebra, and thus we have to invoke a more generalized mathematical definition, say ultrafilters over a partial order (or preorder, almost equivalently). But then what is the partial order?

Thankfully, recall our previous informal discussion on stronger arguments and weaker arguments in Section~\ref{sec.k2a}, there exactly exists a quite handy and natural notion of partial order among arguments, namely by order of argumentation strength. And then immediately, our core philosophical idea bears intimate similarity to the ultrafilter extension construction in modal logic, so that proper filters correspond to ``consistent theories'' while ultrafilters correspond to ``maximal consistent theories'', for readers unfamiliar with this topic, cf. \cite{Blackburn01} as a rudimentary introduction. Formal definitions go as the following:
\begin{definition}[Preorder $\leq$]\label{def.pre}
Given an argument Kripke model $\N=(\W,A,G)$, we define a preorder $\leq\subseteq\W\times\W$ such that for any $U,V\in\W$, $U\leq V$ iff $\forall p\in P\forall W\in\W$, $(W,V)\in A(p)\to(W,U)\in A(p)$.
\end{definition}
Verifying that $\leq$ is reflexive and transitive is obvious. As for the intuitive reading, $U\leq V$ means ``argument $U$ is no weaker than argument $V$ in attack power'', just as Definition~\ref{def.pre} shows, if and only if any pair of proposition and argument attacked by $V$ also gets attacked by $U$.
\begin{remark}
$U\leq V$ intuitively means ``$U$ is stronger than $V$'' (or at least equally strong), \textit{not} ``$V$ is stronger than $U$'', and therefore, here it is inadvisable to read the notation ``$\leq$'' as ``less than or equal to''. This kind of definition in many ways conforms to the convention in set theory \cite{Jech03}: the ``smaller'', the stronger.
\end{remark}
Given any fixed partial order (or more generally, fixed preorder) $\leq$ over any fixed nonempty set $\W$, the usual definitions of filters and ultrafilters \cite{Kunen80} are presented as below:
\begin{definition}[Filter]
For any $\U\subseteq\W$, $\U$ is a filter iff:
\begin{itemize}
    \item $\U\neq\emptyset$.
    \item $\forall U\in\U$, $\forall W\in\W$, $U\leq W\to W\in\U$.
    \item $\forall U,V\in\U$, $\exists W\in\U$, $W\leq U\land W\leq V$.
\end{itemize}
A filter $\U$ is proper iff $\U\neq\W$, i.e., $\U\subset\W$.
\end{definition}
\begin{definition}[Ultrafilter]
For any $\U\subseteq\W$, $\U$ is an ultrafilter iff:
\begin{itemize}
    \item $\U$ is a proper filter.
    \item $\neg\exists\V\subseteq\W(\U\subset\V\land\V$ is a filter).
\end{itemize}
\end{definition}
Here we have to pay a little more attention, however. Unlike set-theoretic ultrafilters, which always exist for any nonempty set $S$ as subsets of $\P(S)$, for a portion of partial orders or preorders there exist no ultrafilters over them. To give a pictorial though informal depiction, those partial orders without ultrafilters never ``diverge permanently'' towards totally independent directions, thus so to speak, they are either just linear orders or in a sense much like linear orders. Fortunately, such kind of partial orders matter little to argumentation theory --- no need to argue if there is no divergence after all! Therefore, as the following Definition~\ref{def.tri}, in this paper we can safely concentrate on nontrivial argument Kripke models with ultrafilters.
\begin{definition}[Triviality]\label{def.tri}
An argument Kripke model $\N$ is trivial iff there is no ultrafilter over the preorder $\leq$ in Definition~\ref{def.pre}. $\N$ is nontrivial iff it is not trivial.
\end{definition}
Through viewing ultrafilters as possible worlds, we at last approach the following Definition~\ref{def.a2k}:
\begin{definition}[Generated Epistemic Kripke Model]\label{def.a2k}
Given a nontrivial argument Kripke model $\N=(\W,A,G)$ with a possible argument $W\in\W$, the generated epistemic Kripke model is $\M_g=(S,E,F)$, where:
\begin{itemize}
    \item $S=\{\U\subseteq\W\mid\U$ is an ultrafilter over the preorder $\leq\}$.
    \item For any $i\in I$ and any $\U,\V\in S$, $(\U,\V)\in E(i)$ iff $\U\cap G(i)=\V\cap G(i)$.
    \item For any $p\in P$, $F(p)=\{\U\in S\mid\exists U\in\U\forall V\in\U,V\leq U\to(V,W)\notin A(p)\}$.
\end{itemize}
\end{definition}
It is easy to see that $S\neq\emptyset$ and that $E(i)$ is an equivalence relation, thus $\M_g$ in Definition~\ref{def.a2k} is indeed a well-defined epistemic Kripke model. Readers may now wonder, however, that we have yet provided no justification for $E(i)$ or $F(p)$ in Definition~\ref{def.a2k}. As a matter of fact, since possible worlds are constituted of ultrafilters, which are somewhat mathematically factitious after all, it is inevitably tough to give further intuitive explanation solely based on our folk experience. Rather, we shall adopt a little indirect but strictly formal approach, through investigating into a beautiful duality theorem on invariance over the language in Definition~\ref{def.lan}, so as to demonstrate that our generations in Definitions \ref{def.k2a} and \ref{def.a2k} are indeed perfectly reasonable. Nevertheless, before we dive into technical subtleties related to this duality theorem in Section~\ref{sec.dua}, we had better first take a look at one vivid practical example presented in the following Section~\ref{sec.exa}, both for acquainting readers with our rather abstract mathematical definitions above and for illustrating the significant practicality of such two-way generation as well as our framework as a whole.

\section{A Practical Example: Running for the Mayor}\label{sec.exa}

\begin{example}\label{exa.may}
Alice and Bob, two of the mayor candidates, are publicly debating over their intended urban policies.

`Aging is the biggest challenge our city is now facing,' says Alice. `There are more and more old people, yet young population is decreasing. If I am elected as the mayor, I will wield a series of revolutionary measures that attract young people to move and dwell into our city.'\hfill$\cdots\cdots(A_1)$

`No, we don't need any more young people who will end up competing for only a limited number of job positions,' says Bob. `In my opinion, it is unemployment that constitutes our biggest problem, since at present around 1,000 young people are badly in need of a job. Therefore if I am elected as the mayor, I promise to create 1,000 new job opportunities during 5 years.'\hfill$\cdots\cdots(B)$

`I cannot agree with you,' says Alice. `As far as I know, our city's working-age population is rapidly decreasing at a current speed of around minus 400 per year, so what will be the point of those 1,000 more job opportunities if there are merely not enough young people who can take them?'\hfill$\cdots\cdots(A_2)$
\end{example}
In actuality, the above debate depicted in Example~\ref{exa.may} may well continue going on, but to utter a \textit{real-time} commentary on the debate, at the present moment Alice's last argument attacks Bob's, and thus based on our folk principle that ``the one who has the last word laughs best'', we can conclude that Alice is \textit{tentatively} winning the debate. Therefore, we can formalize this debate as the following:
\begin{itemize}
    \item Two agents $a,b$ denote Alice and Bob, respectively.
    \item $A_1$, $B$, $A_2$ denote the three consecutive arguments by Alice and Bob, respectively, as we have indexed beside the text above in Example~\ref{exa.may}. Our current possible argument is the last argument, namely $A_2$.
    \item Proposition $p$ denotes that ``there exists considerable unemployment.'' Argument $B$ attacks argument $A_1$ with respect to proposition $p$.
    \item Proposition $q$ denotes that ``considering aging, unemployment is not a serious problem''. Argument $A_2$ attacks argument $B$ with respect to proposition $q$.
    \item At present, arguments $A_1$ and $A_2$ are of course available to Alice herself. Nonetheless, argument $B$ is \textit{not} available to Bob \textit{now} (although it used to be), because it is losing to the current argument $A_2$ which is attacking it and hence Bob has to give it up.
\end{itemize}
Then, our current argument Kripke model $\N$ should look like the following:
$$\xymatrix@=50pt{
\N & A_1:a,\neg b\ar[r]|{A(p)} & B:\neg a,\neg b\ar[r]|{A(q)} & A_2:a,\neg b
}$$
Based on model $\N$, the preorder $\leq$ in Definition~\ref{def.pre} can be represented as the following Hasse diagram, namely $A_2\leq A_1$ and $B\leq A_1$:
$$
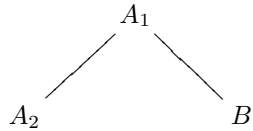
\xymatrix{
& A_1 &\\
A_2\ar@{-}[ur] & & B\ar@{-}[ul]
}$$
And so obviously there exist two ultrafilters, namely $\{A_1,A_2\}$ and $\{A_1,B\}$.

Before finalizing our construction of the generated epistemic Kripke model $\M_g$, let us pause for a while and genuinely think about what we are actually doing now and why we are doing it. Surely, we desire a familiar, tidy and expressive epistemic Kripke model that can aid us to reason about the current scenario, as epistemic logic has been thoroughly studied over the decades and has proved to be a quite accurate and productive depiction of agents' knowledge, especially helpful when dealing with intricate situations such as the muddy children puzzle and the Byzantine generals problem \cite{Fagin95}. Nonetheless, there remains one fatal concern. Indeed, the common method to obtain epistemic Kripke models in reality still relies largely on human intuition and so is rather obscure, a deficiency which in fact greatly impedes such kind of logic approach to knowledge presentation from being universally put into practical application \cite{Perrotin19}. For instance, just under our current setting of the above Example~\ref{exa.may}, we have already formalized the argument Kripke model $\N$ with ease, simply as a literal transcription of the debate \textit{per se}, however, the task of generating the epistemic Kripke model for Alice and Bob is far from obvious or even natural. After all, in the realistically physical world, it is solely Alice and Bob's \textit{extrinsic} debate that is plainly accessible to other audience, whereas their \textit{intrinsic} knowledge could be indirectly inferred but never directly observed.

The above analysis accounts for exactly why we are about to compute the generated epistemic Kripke model $\M_g$ from our original argument Kripke model $\N$, and also illuminates one of the major motivations as well as applications of our dual framework presented in this paper. Hence as defined in the previous Section~\ref{sec.a2k}, the generated epistemic Kripke model $\M_g$ from the argument Kripke model $\N$ together with the current possible argument $A_2$ is calculated to be the following:
$$\xymatrix@=50pt{
\M_g & \{A_1,A_2\}:p,q\ar@(ur,ul)@{.}|{E(a),E(b)}\ar@{.}[r]|{E(b)} & \{A_1,B\}:p,\neg q\ar@(ur,ul)@{.}|{E(a),E(b)}
}$$
It can be noted that $\M_g$ contains two possible worlds, namely the two ultrafilters $\{A_1,A_2\}$ and $\{A_1,B\}$. Thus which one should be our current possible world? The current possible world should of course include the current possible argument $A_2$, and so it must be the ultrafilter $\{A_1,A_2\}$, where both $p$ and $q$ really hold. Now as $\M_g$ also shows, Alice actually knows the fact $q$ but Bob does not, since he also considers another world (i.e. ultrafilter) $\{A_1,B\}$ possible, where $\neg q$ holds instead. Hence in a word, our resulting epistemic Kripke model $\M_g$, albeit automatically generated from the argument Kripke model $\N$, indeed conforms pretty well with what a human being would intuitively deduce from the argumentation process between Alice and Bob in Example~\ref{exa.may}.

Admittedly, in this specific simple example, $\M_g$ looks so obvious just by our intuition, and some of the readers might still be wondering why we essentially need such an indirect process to generate $\M_g$ rather than try to directly construct the current epistemic Kripke model. But here the main point is, no matter how simple or how complex $\M_g$ may actually be, to say the very least, any straightforward clues about what the epistemic Kripke model is like can be scarcely found either in the real physical world, or in the narration of some story happened in real world. For instance considering the muddy children puzzle, an ordinary solution would commonly start with a cubic epistemic Kripke model which is almost taken for granted, although this puzzle's literal text --- namely some English sentences describing the background settings of the puzzle followed by the conversation between the parent and the children --- actually shows no direct hint for any cubic structure at all. As a matter of fact, according to David Lewis' modal realism \cite{Lewis01}, all possible worlds in the epistemic Kripke model are equally real, which could suggest that since we are at our current world, we should not be fully capable of observing other possible worlds as directly as the current world, thus not for the whole epistemic Kripke model, either. Therefore in argumentation theory as well as other fields of application, quite substantially, the \textit{intangible} epistemic Kripke model is always generated from other concrete \textit{observables}, such as the argument Kripke model. Even though in relatively easy cases like Example~\ref{exa.may} and the muddy children puzzle, the epistemic Kripke model seems fairly apparent, nonetheless, human's intuition may fail to function well when confronted with much complicated circumstances, and moreover importantly, such vague concept of human intuition possesses little reference value for AI systems to be counted as an effective algorithm. Hence, our approach by the dual theory helps to fill up the above gap through automatically generating the appropriate epistemic Kripke model in a systematical way.

As a conclusion to our philosophical analysis over Example~\ref{exa.may} in this section, we have carefully demonstrated how to generate the corresponding epistemic Kripke model $\M_g$ given the current argument Kripke model $\N$. In the next Section~\ref{sec.dua}, we shall then turn our attention to some technical details and in the end prove a duality theorem so as to formally justify our such methodology of two-way generation.

\section{Knowledge vs. Argumentation: The Duality}\label{sec.dua}

Fix an epistemic Kripke model $\M=(S,E,F)$ with a possible world $s\in S$, and suppose its corresponding generated argument Kripke model is $\N_g=(\W,A,G)$, where we have $\{s\}\in\W$. Assume $\N_g$ is nontrivial. Thus once again, with respect to this fixed argument Kripke model $\N_g$ and this fixed possible argument $\{s\}$, suppose its corresponding generated epistemic Kripke model is $\M_{gg}=(S',E',F')$. We reasonably expect that $\M_{gg}$, even though may not be exactly identical to the original $\M$, must resemble $\M$ to a certain extent. Nonetheless, in order to pin down such intuition precisely, we firstly have to determine which possible world in $S'$ should correspond to the original $s\in S$, namely, represent the currently real world.
\begin{lemma}\label{lem.pu}
For any $t\in S$ namely $\{t\}\in\W$, there does not exist $U\in\W$ such that $U\leq\{t\}\land\{t\}\nleq U$, where relation $\leq$ is defined in Definition~\ref{def.pre}.
\end{lemma}
\begin{proof}
To a contradiction suppose such $U$ exists, since $\{t\}\nleq U$, there exists $p\in P$, $V\in\W$ such that $(V,U)\in A(p)\land(V,\{t\})\notin A(p)$, hence there does not exist $v\in V$ such that $(v\in F(p)\land t\notin F(p))\lor(v\notin F(p)\land t\in F(p))$, namely for any $v\in V$, $v\in F(p)\iff t\in F(p)$. And then as $(V,U)\in A(p)$, it is easy to see that for any $u\in U$, $v\in V$, $u\in F(p)\iff v\notin F(p)\iff t\notin F(p)$, thus $(U,\{t\})\in A(p)$ but of course $(U,U)\notin A(p)$, therefore $U\nleq\{t\}$, a contradiction.
\end{proof}
\begin{proposition}\label{pro.pu}
For any $t\in S$, $\tau(t)=\{U\in\W\mid\{t\}\leq U\}\subseteq\W$ is an ultrafilter.
\end{proposition}
\begin{proof}
Since $\{t\}\leq\{t\}$, $\{t\}\in\tau(t)$ so $\tau(t)\neq\emptyset$. For any $U\in\tau(t)$, any $W\in\W$, if $U\leq W$ then by transitivity $\{t\}\leq W$, hence $W\in\tau(t)$. For any $U,V\in\tau(t)$, $\exists\{t\}\in\tau(t)$ such that $\{t\}\leq U\land\{t\}\leq V$. Thus $\tau(t)$ is a filter. To a contradiction suppose $\tau(t)=\W$, then any proper filter $\U\subset\W$ is not an ultrafilter, namely there are no ultrafilters, contradicting our assumption that $\N_g$ is nontrivial. Thus $\tau(t)\subset\W$ is a proper filter. To a contradiction suppose there exists $\U\subseteq\W$ such that $\tau(t)\subset\U$ and that $\U$ is a filter, then there exists $U\in\U\setminus\tau(t)$, as $\U$ is a filter, there exists $V\in\U$ such that $V\leq U\land V\leq\{t\}$, and because $U\notin\tau(t)$, we also have $V\notin\tau(t)$, namely $\{t\}\nleq V$, contradicting the previous Lemma~\ref{lem.pu}. Thus $\tau(t)$ is an ultrafilter.
\end{proof}
\begin{definition}[Principal Ultrafilter]
As Proposition~\ref{pro.pu} shows, for any $t\in S$, we call $\tau(t)\in S'$ as $t$'s principal ultrafilter.
\end{definition}
Although the definition here is slightly different from ordinary principal ultrafilters over a Boolean algebra \cite{Hodges97}, readers should feel convinced that the core idea is the same and so we are quite justified for usage of the same name. Then of course, the current possible world in the ``generated generated'' epistemic Kripke model $\M_{gg}$ should be $\tau(s)\in S'$, namely, the principal ultrafilter of the original current world $s\in S$.

A little more preparation is required: we additionally assume that for any proposition $p\in P$, $\M,s\vDash p$ i.e. $s\in F(p)$. The reason is that an argument Kripke model can only inform us about the \textit{relative} attack relations between arguments with respect to some proposition $p$, but nothing concerned with their \textit{absolute} true-hood or falsehood, hence such kind of assumption, or to be called more accurately as convention or protocol, is necessarily needed. Anyhow, this restriction actually does not matter a lot, since due to symmetry between true and false we can always freely choose either $p$ or $\neg p$ as the proper basic proposition so as to meet the assumption's requirements. We then firstly prove a series of useful lemmata under our current assumption:
\begin{proposition}
For any $U,V\in\W$ such that $U\subseteq V$, $U\leq V$.
\end{proposition}
\begin{proof}
This fact can be verified straightforwardly from definition.
\end{proof}
\begin{lemma}\label{lem.bas}
For any $t\in S$ and any $p\in P$, $\M_{gg},\tau(t)\vDash p\iff\M,t\vDash p$.
\end{lemma}
\begin{proof}
We have
\begin{align*}
\M_{gg},\tau(t)\vDash p\iff & \tau(t)\in F'(p)\\
\iff & \exists U\in\tau(t)\forall V\in\tau(t),V\leq U\to(V,\{s\})\notin A(p)\\
\iff & \exists U\in\tau(t)\forall V\in\tau(t),V\leq U\to\neg\exists r\in V\forall s\in\{s\},\\
& (r\in F(p)\land s\notin F(p))\lor(r\notin F(p)\land s\in F(p))\\
\iff & \exists U\in\tau(t)\forall V\in\tau(t),V\leq U\to\forall r\in V,\\
& r\in F(p)\leftrightarrow s\in F(p)\\
\iff & \exists U\in\tau(t)\forall V\in\tau(t),V\leq U\to\forall r\in V,r\in F(p) & (\dag)
\end{align*}
On the one hand, suppose (\dag) holds, then since $U\in\tau(t)$, we have $\{t\}\leq U$, so $\forall r\in\{t\},r\in F(p)$, hence $t\in F(p)$. On the other hand, suppose $t\in F(p)$, then $\{t\}\in\tau(t)$, and for any $V\in\tau(t)$ such that $V\leq\{t\}$, namely for any $V\in\tau(t)$ such that $\forall q\in P\forall W\in\W$, $(W,\{t\})\in A(q)\to(W,V)\in A(q)$, we thus have $(W,\{t\})\in A(p)\to(W,V)\in A(p)$. If $\neg\exists u\in S$ such that $u\notin F(p)$, then of course $\forall r\in V\subseteq S$, $r\in F(p)$. If $\exists u\in S$ such that $u\notin F(p)$, then $(\{u\},\{t\})\in A(p)$ thus $(\{u\},V)\in A(p)$, therefore we still have $\forall r\in V$, $r\in F(p)$. Anyway, (\dag) always holds. In total, we have $(\dag)\iff t\in F(p)$, namely $\M_{gg},\tau(t)\vDash p\iff\M,t\vDash p$.
\end{proof}
\begin{lemma}\label{lem.zig}
For any $t\in S$ and any $i\in I$, if $(s,t)\in E(i)$ then $(\tau(s),\tau(t))\in E'(i)$.
\end{lemma}
\begin{proof}
For any $U\in\tau(s)\cap G(i)$ we have $t\in U$, hence $\{t\}\leq U$, so also $U\in\tau(t)\cap G(i)$, and vice versa, thus $\tau(s)\cap G(i)=\tau(t)\cap G(i)$, namely $(\tau(s),\tau(t))\in E'(i)$.
\end{proof}
\begin{lemma}\label{lem.zag}
For any $i\in I$, $p\in P$ and any $\U\in S'$ such that $(\tau(s),\U)\in E'(i)$,
\begin{enumerate}
    \item if $\M_{gg},\U\vDash p$ then there exists $t\in S$ such that $(s,t)\in E(i)$ and that $\M,t\vDash p$.
    \item if $\M_{gg},\U\nvDash p$ then there exists $t\in S$ such that $(s,t)\in E(i)$ and that $\M,t\nvDash p$.
\end{enumerate}
\end{lemma}
\begin{proof}
Denote $[s]_i=\{t\in S\mid(s,t)\in E(i)\}\in\W$. Since $s\in[s]_i$, $\{s\}\leq[s]_i$ namely $[s]_i\in\tau(s)$, also since $[s]_i\in G(i)$, we have $[s]_i\in\tau(s)\cap G(i)$. As $(\tau(s),\U)\in E'(i)$, namely $\tau(s)\cap G(i)=\U\cap G(i)$, hence $[s]_i\in\U$.

Moreover the same as (\dag) in Lemma~\ref{lem.bas}, we have $\M_{gg},\U\vDash p\iff\exists U\in\U\forall V\in\U,V\leq U\to\forall r\in V,r\in F(p)$.
\begin{enumerate}
    \item If $\M_{gg},\U\vDash p$, namely $\exists U\in\U\forall V\in\U,V\leq U\to\forall r\in V,r\in F(p)$, then as $[s]_i\in\U$ and $\U$ is a filter, there exists $W\in\U$ such that $W\leq[s]_i\land W\leq U$, hence $\forall r\in W$, $r\in F(p)$, and because $W\neq\emptyset$ we have $\exists r\in W$, $r\in F(p)$. Since $W\leq[s]_i$ and $(W,W)\notin A(p)$, we have $(W,[s]_i)\notin A(p)$, therefore $\exists t\in[s]_i$ such that $t\in F(p)$, namely $(s,t)\in E(i)$ and $\M,t\vDash p$.
    \item If $\M_{gg},\U\nvDash p$, namely $\forall U\in\U\exists V\in\U,V\leq U\land\exists r\in V,r\notin F(p)$, then as $[s]_i\in\U$, there exists $V\in\U$ such that $V\leq[s]_i\land\exists r\in V$, $r\notin F(p)$. Since $V\leq[s]_i$ and $(V,V)\notin A(p)$, we have $(V,[s]_i)\notin A(p)$, therefore $\exists t\in[s]_i$ such that $t\notin F(p)$, namely $(s,t)\in E(i)$ and $\M,t\nvDash p$.
\end{enumerate}
\end{proof}
Finally we state our main result as the following duality theorem:
\begin{theorem}[Duality Theorem]\label{the.dua}
For any epistemic formula $\phi$ where epistemic modalities only act upon atomic propositions and their negations, we have $\M,s\vDash\phi\iff\M_{gg},\tau(s)\vDash\phi$.
\end{theorem}
\begin{proof}
The atomic case when $\phi$ is the form of $p$ follows from Lemma~\ref{lem.bas}. The modal case when $\phi$ is the form of $\K_i p$ or $\K_i\neg p$ follows from Lemmata \ref{lem.bas} and \ref{lem.zig} for one direction, as well as Lemma~\ref{lem.zag} for another direction.
\end{proof}
It may seem a little pity that according to Theorem~\ref{the.dua}, only agents' knowledge about atomic propositions (and their negations) is ensured to be preserved from $\M$ to $\M_{gg}$, but not necessarily for more complex knowledge like $\K_i(p\to q)$ or mutual knowledge like $\K_i\K_j p$. However, such limitation is exactly due to the fact that in Definition~\ref{def.akm} of argument Kripke models, possible arguments can attack each other only with respect to each single atomic proposition, but not with respect to a set of propositions like $\{p,q\}$ or agents' knowledge like $\K_jp$. For simplicity and clarity, we choose to work with this restricted version of definition in the major part of this paper, but shall also briefly discuss possible directions to expand the model and semantics in the following Subsection~\ref{sub.sem} so that we can obtain richer representation and finer-grained resolution.

On the other hand symmetrically, it is also totally feasible to start from an argument Kripke model $\N=(\W,A,G)$ with a possible argument $W\in\W$ and end up with the ``generated generated'' argument Kripke model $\N_{gg}$, while a corresponding duality theorem similar to Theorem~\ref{the.dua} should then be expected under a suitable set of presumptions. Anyhow, this dual generating process from $\N$ to $\N_{gg}$ might seem less intuitive or practically useful, and we shall omit detailed analysis here since by duality, the mathematical technique is just very similar.

In sum, this kind of duality theorem formally justifies our two-way generation in Definitions \ref{def.k2a} and \ref{def.a2k}, so that the \textit{static}, \textit{intrinsic} epistemic Kripke models and the \textit{dynamic}, \textit{extrinsic} argument Kripke models are indeed a dual pair.

\section{Further Discussion}\label{sec.dis}

This section collects a bunch of miscellaneous discussion on further expansions and applications of our framework presented in this paper.

\subsection{Syntax Expanded}\label{sub.syn}

In Definition~\ref{def.lan}, a symmetric pair of languages, namely the epistemic formula $\phi$ and the argument formula $\xi$, are defined as the ordinary basic modal logic. While epistemic logic has been well established in the literature, readers may keep wondering about the practical usefulness of the argument formulae, in spite of our explanation for their intuitive readings right after Definition~\ref{def.lan}. As a matter of fact, argument formulae are just as expressive as epistemic formulae, both mathematically and pragmatically, that is to say, they are indeed able to talk about lots of interesting features of argumentation. As an elementary example, the following very simple argument formula $\A_p\neg i$ intuitively says that, all arguments that attack the current argument with respect to proposition $p$ are not accessible to agent $i$, in other words, the current argument ``defeats'', ``dominates'' or ``overrules'' agent $i$ with respect to proposition $p$ because he has no available counter arguments to refute it, hence, whoever utters this argument out will certainly win the debate over proposition $p$ against agent $i$. Although we admit that the usage itself might seem a little bit awkward --- probably because it somewhat differs from our familiar SVO word order in English --- substantially, argument formulae are competently powerful for expressing a variety of situations in argumentation.

Moreover, we mainly restrict our attention in this paper onto basic modal logic for the reason of simplicity and clarity, but of course we by no means have to take such limit rigidly. Since sole argument formulae are definitely not enough for distinguishing every detail of argument Kripke models, we are rather free to expand the language, for example by introducing new modalities, so as to obtain more abundant expressivity just in any case of need.

\subsection{Semantics Expanded}\label{sub.sem}

On the other hand, it is also plausible to consider expanding the Kripke models, especially the argument one, with more affluent structure and information. As Definition~\ref{def.akm} shows, for simplicity at present, we only allow arguments to attack each other with respect to certain atomic proposition. However let alone in mathematics, even in real life, people often argue for different opinions which do not disagree simply on one atomic proposition. For example, two agents do not know whether Alice loves Bob or Bob loves Alice, hence there is no divergence on either atomic proposition, but one agent believes that Alice loves Bob if and only if Bob loves Alice, while the other agent does not think it is necessarily so. Furthermore, arguments can also disagree over agents' knowledge, such as $\K_jp$. Through taking various kinds of attack relations into consideration we will then be able to drop the restrictions on language in Theorem~\ref{the.dua}, as we have remarked right after that theorem's proof.

For another common concern on the relationship between our dual framework presented in this paper and Dung's argumentation theory \cite{Dung95}, as we have foretold in Section~\ref{sec.k2a}, Dung's arguments seem to attack each other only in the weak sense, because accordingly in his definition of consistency, any two arguments in a consistent set can never attack one another, which seems to be even stricter than our later definition in Section~\ref{sec.a2k} of maximal consistent sets of arguments as ultrafilters over a preorder. Anyway, although Dung's definition is perhaps not strong enough to fulfill our demand here as a dual theory, we might still consider using it in certain due case.

On the flip side, not only attack relations, but support relations between arguments may also get dealt with as another sort of modality. In total, the same technique applies to all these semantic expansions with various kinds of relations amongst the arguments.

\subsection{Application to Game Theory}

We would also like to suggest that our framework is versatile and flexible for applying to game theory, which contributes as another quite important practical usage. For a very rough instance, agents who are game players might take turns to pick one possible argument that is currently available to himself from the argument Kripke model, and as explained in the previous Subsection~\ref{sub.syn}, agent $i$ will lose this game of argumentation over proposition $p$ against any argument on which $\A_p\neg i$ holds. Briefly speaking, in accordance with different kinds of practical scenarios, our detailed game rules can also be set differently, and moreover, agents can devise their personal game strategies based on their own knowledge from the epistemic Kripke model.

\subsection{Dynamic Essence of Knowledge as Maxwell's Demon}

Commonly, knowledge is regarded as some \textit{static} state of the objective world, for example, we talk about the agents' knowledge at a single moment. The traditionally accepted definition of knowledge as justified true belief or something else alike does not seem to reveal any temporal characteristic of knowledge, either \cite{Audi98}. Nonetheless, our approach in this paper to treat knowledge and argumentation as a dual creatively provides an alternative angle of philosophical view, asserting that knowledge also bears an essential \textit{dynamic} dimension, as the current epistemic Kripke model is always accompanied by its corresponding generated argument Kripke model.

As for an intuitive explanation, here we would like to resort to Maxwell's demon. Based on both thermodynamics \cite{Reif65} and information theory \cite{Cover06}, the recognized explication for how the demon works is that the demon has to know each molecule's velocity in order to perform his task, while such knowledge (or to be called information) itself is negative entropy, in other words in the physical reality, entropy must increase when the demon attempts to acquire new knowledge about a molecule's velocity by whatever means. Hence the demon's seemingly \textit{subjective} knowledge as a \textit{static} mental state actually carries physically \textit{objective} negative entropy, exactly because the demon is able to act \textit{dynamically} according to his knowledge and thus influence the physical world. Therefore, this well known thought experiment strongly defenses our metaphysical and epistemological thesis that knowledge is essentially dynamic.

\section{Concluding Remarks}\label{sec.con}

In this paper, we come up with a dual theory between static knowledge and dynamic argumentation so as to uncover the two concepts' profound relevance. Particularly, we define in symmetry two types of modal formulae, for knowledge and argumentation respectively, as well as two corresponding types of Kripke models. We then \textit{semantically} establish the two-way generation method capable of transforming either type of Kripke model to the other type, and also rigidly justify our methodology by showing the \textit{syntactic} invariance via a dual theorem. Our novel yet convincing approach to combine knowledge and argumentation fundamentally gets rid of any unnatural affectations in other previous attempts, while our simple yet clear framework is supposed to lay out a basic playground for both epistemic logic and argumentation theory.

Moreover, through careful examination on a vivid example of pragmatically utilizing our dual theory, we not only demonstrate in a concrete way how the methodology of two-way generation actually works, but we are also led to the eventual discovery of our framework's one significantly useful application in practice: realistically, the argument Kripke model is usually much easier to observe and obtain, for example simply by recording the actual verbal debate among the agents, from the source data of which we are then able to generate the corresponding epistemic Kripke model and carry out a real-time analysis over the current debate. What is more important, the above task essentially relies little on human intuition and so is also viable to get mechanically performed by an AI system, who can properly implement modal logic's decidability as one of its main advantages.

Further, we discuss possible expansions as well as applications of our framework, which may constitute some future research directions. Especially, we claim that in contrary to our common perception, the seemingly ``\textit{static}'' knowledge essentially possesses an indispensably \textit{dynamic} instinct. As Francis Bacon's famous saying goes: ``\textit{Knowledge is power}''. Such admitted characterization of knowledge also displays a vibrant image to help us intuitively understand knowledge's dynamic essence, as power can only manifest itself through actual exertion of the power in some specific dynamic event, for instance, by demonstrating your argument power to skillfully defeat your opponent during a debate, while your such power exactly comes out of your knowledge. All the entanglement between statics and dynamics reminds us of Karl Marx's extremely influential motto: ``The philosophers have hitherto only \textit{interpreted} the world in various ways; the point, however, is to \textit{change} it''.

\section*{Acknowledgements}

Both the authors owe immense gratitude to Teeradaj Racharak (whose nickname is `X') for his teaching and help on argumentation theory. We would also like to thank the anonymous reviewers for their invaluable advice.

\clearpage

\bibliographystyle{splncs04}
\bibliography{main}

\end{document}